\PassOptionsToPackage{table}{xcolor}
\documentclass[11pt]{article}
\usepackage[english]{babel}
\usepackage{microtype}
\usepackage{graphicx}
\usepackage{array}
\usepackage{fullpage}
\usepackage{appendix}
\usepackage{bbold}
\usepackage[T1]{fontenc}
\usepackage{lmodern}
\usepackage{csquotes}
\usepackage[dvipsnames]{xcolor}
\usepackage{latexsym}
\usepackage[intlimits]{amsmath}
\usepackage{amsthm}
\usepackage{amsxtra}
\usepackage{amssymb}
\usepackage{booktabs}
\usepackage{makecell}
\usepackage{tabularx}
\usepackage{diagbox}
\usepackage{amsfonts}
\usepackage{mathtools}
\usepackage{centernot}
\usepackage{natbib}
\usepackage{xurl}
\usepackage{aliascnt}
\usepackage[colorlinks=true,linktoc=all,linkcolor=blue,citecolor=purple,urlcolor=red, backref=page]{hyperref}

\theoremstyle{plain}
\newtheorem{theorem}{Theorem}
\newtheorem{lemma}{Lemma}[section]

\newaliascnt{claim}{lemma}
\newtheorem{claim}[claim]{Claim}
\aliascntresetthe{claim}

\newtheorem{proposition}[lemma]{Proposition}

\newtheorem{definition}[lemma]{Definition}
\newtheorem{remark}[lemma]{Remark}

\DeclareRobustCommand{\fullref}[1]{%
  \hyperref[#1]{\autoref*{#1}}%
}

\usepackage{thm-restate}

\newcommand{\gch}{\textcolor{Green}{\checkmark}}
\newcommand{\rch}{\textcolor{RedOrange}{\checkmark}}
\newcommand{\efonepo}{$\mathrm{EF1}+\mathrm{PO}$}
\newcommand{\efkpo}{$\mathrm{EF}k+\mathrm{PO}$}
\newcommand{\NP}{$\mathrm{NP}$}

\def\abs#1{\left| #1 \right|}

\newcommand{\bvc}{\ensuremath{\textsc{Balanced Vertex Cover}}}

\newcommand{\vcov}{\ensuremath{v_{\operatorname{cov}}}}

\newcommand{\ignore}[1]{}

\AtBeginDocument{%
}

\title{Approximate Envy-Freeness and Pareto Optimality Beyond Additivity}

\author{%
   Harish Chandramouleeswaran\\
    Chennai Mathematical Institute, India\\
    \texttt{harishc@cmi.ac.in} \\
    \and
    Prajakta Nimbhorkar\\
    Chennai Mathematical Institute, India\\
    \texttt{prajakta@cmi.ac.in}
}

\date{}

\begin{document}

\maketitle

\begin{abstract}
We study envy-freeness up to one good ($\mathrm{EF1}$) together with Pareto optimality ($\mathrm{PO}$) for indivisible goods. For two agents with monotone submodular valuations, we prove that a $\frac{2}{3}$-$\mathrm{EF1}$+$\mathrm{PO}$ allocation always exists, while for every $\alpha > \frac{3}{4}$, an $\alpha$-$\mathrm{EF1}$+$\mathrm{PO}$ allocation may fail to exist even for two agents with unweighted coverage valuations. A six-good $\mathrm{EF1}$+$\mathrm{PO}$ counterexample resolves the longstanding exact existence question \cite{CaragiannisKMPSW19,BenabbouCIZ21,BarmanS26} negatively.

Exact $\mathrm{EF1}$+$\mathrm{PO}$ does exist when one agent has a monotone submodular valuation and the remaining agents have identical additive valuations. Finally, we consider the compatibility of the weaker fairness notion of $\mathrm{EF}k$ with $\mathrm{PO}$. For two agents, we establish tight $\mathrm{EF}k$+$\mathrm{PO}$ item thresholds of $4k+1$, $2k+1$, and $3k+1$ for two submodular agents, two subadditive agents, and one subadditive and one additive agent, respectively. Beyond these guarantees, deciding $\mathrm{EF}k$+$\mathrm{PO}$ existence is $\mathrm{NP}$-hard for every fixed $k\geq 1$, even for three agents with unweighted coverage valuations.
\end{abstract}


\section{Introduction}

Fair division studies how scarce resources should be allocated when agents have heterogeneous preferences. Two fundamental outcomes that are desired from an allocation of resources among the agents are \emph{fairness}, which reduces the disparity (or `envy') among agents, and \emph{efficiency}, which rules out reallocations that benefit an agent without harming any other agent. \emph{Envy-freeness} ($\mathrm{EF}$) \cite{Foley66} is one of the most natural notions of fairness, while \emph{Pareto optimality} ($\mathrm{PO}$) \cite{Little50} is a natural notion of (economic) efficiency. In the classical setting of \emph{divisible} items (like a cake), envy-freeness and Pareto optimality can often be simultaneously satisfied \cite{Foley66,Varian74}. The case of \emph{indivisible} goods is already different -- exact envy-freeness fails even for the simplest instance of two agents and one good which is valuable for both agents. Hence several relaxations of envy-freeness have been studied in this setting.

One of the most popular and well-studied relaxations of envy-freeness in the case of indivisible goods is the notion of \emph{envy-freeness up to one item} ($\mathrm{EF1}$) \cite{Budish11}. An allocation is $\mathrm{EF1}$ if every agent's envy of another bundle can be eliminated by removing at most one item from the envied bundle. $\mathrm{EF1}$ allocations always exist when agents have arbitrary monotone valuations and can be computed in polynomial time by the Envy-Cycle-Elimination algorithm \cite{LiptonMMS04}. Pareto optimality ($\mathrm{PO}$), on the other hand, requires that no reallocation weakly improves every agent, and strictly improves at least one.

For additive goods, maximizing Nash social welfare yields allocations that are both $\mathrm{EF1}$ and $\mathrm{PO}$ \cite{CaragiannisKMPSW19}. The problem of computing an allocation that maximizes the Nash Social Welfare is, however, not only $\mathrm{NP}$-hard \cite{NguyenNRR14}, but also $\mathrm{APX}$-hard \cite{Lee17}. There have since been significant algorithmic results for computing $\mathrm{EF1}+\mathrm{PO}$ allocations by bypassing Nash Social Welfare. A pseudopolynomial-time algorithm for computing $\mathrm{EF1}+\mathrm{PO}$ allocations for additive goods is given in \cite{BarmanKV18}.
When the number of  agents is a fixed constant, \cite{Mahara25} gives a polynomial-time algorithm for computing $\mathrm{EF1}+\mathrm{fPO}$ allocations, where fractional Pareto optimality (or $\mathrm{fPO}$) is a notion that is stronger than $\mathrm{PO}$.

This Nash-welfare argument does not extend to general monotone submodular valuations. In the same work \cite{CaragiannisKMPSW19}, it is proved that maximum Nash welfare allocations could fail to satisfy $\mathrm{EF1}$ even for monotone submodular valuations; their positive guarantee is \emph{marginal}-$\mathrm{EF1}$ ($\mathrm{MEF1}$) together with $\mathrm{PO}$. Under monotonicity, $\mathrm{MEF1}$ implies $\frac{1}{2}$-$\mathrm{EF1}$, yielding a $\frac{1}{2}$-$\mathrm{EF1}$+$\mathrm{PO}$ allocation for monotone submodular valuations. They also proved that $\mathrm{EF1}+\mathrm{PO}$ allocations are not guaranteed to exist for monotone subadditive valuations (a strict superclass of monotone submodular valuations), and explicitly left the submodular case open. Subsequently, in \cite{BenabbouCIZ21}, it was proved that $\mathrm{EF1}+\mathrm{PO}$ allocations always exist for matroid-rank valuations, which is a strict subclass of monotone submodular valuations in which every good has binary marginal utility. In a recent work \cite{BarmanS26}, the authors showed that $\mathrm{EF1}$+$\frac{1}{2}$-$\mathrm{PO}$ allocations always exist for monotone subadditive valuations.\footnote{An allocation $X$ is said to be $\frac{1}{2}\text{-}\mathrm{PO}$ if there exists no allocation $Y$ such that for every agent $i$, $\frac{1}{2}v_i(Y_i)\geq v_i(X_i)$ and the inequality is strict for at least one agent \cite{BarmanS26}.}

These Nash-welfare arguments do not directly extend to \emph{chores} (items with nonpositive marginal utility) or to mixtures of goods and chores, known as \emph{mixed manna}. For instances with only two agents with additive valuations over mixed manna, \cite{AzizCIW22} contains a polynomial-time algorithm for computing $\mathrm{EF1}+\mathrm{PO}$ allocations. In a recent breakthrough, \cite{Mahara26} proved that $\mathrm{EF1}+\mathrm{fPO}$ allocations always exist for instances with any number of agents with additive valuations over chores. The result in \cite{Mahara26} also extends to weighted envy-freeness up to one chore ($\mathrm{WEF1}$) for unequal entitlements. \cite{TehSAGT26} computes $\mathrm{WEF1}$+$\mathrm{fPO}$ allocations in polynomial time when the number of chores exceeds the number of agents by at most a fixed constant, allowing arbitrarily many agents and distinct cost functions. For chores with general monotone costs, $\mathrm{EF1}$ allocations always exist and can be computed in polynomial time~\cite{BhaskarSV21}. However, it is known that $\mathrm{EF1}$ can be incompatible with Pareto optimality beyond additive costs. In \cite{HosseiniNW25}, the authors exhibit a two-agent instance with identical monotone submodular (indeed, coverage) cost functions that admits no allocation that is $\mathrm{EF1}+\mathrm{PO}$. Their costs arise from a rooted tree: the cost of a bundle is the size of the minimal subtree connecting its chores to the root. They further show that deciding whether such an $\mathrm{EF1}+\mathrm{PO}$ allocation exists is $\mathrm{NP}$-hard, even for unweighted tree instances, with a variable number of agents.


\subsection{Our Contributions}

The existence of $\mathrm{EF1}+\mathrm{PO}$ allocations under monotone
submodular valuations has remained open
\cite{CaragiannisKMPSW19,BenabbouCIZ21,BarmanS26}.
We resolve this question negatively in \fullref{sec:nonexistence},
even for two agents and six goods.

\begin{restatable}{theorem}{EFonePONonexistence}
\label{thm:EF1PO-nonexistence}
There exists an instance with two agents and six goods
with monotone submodular valuations that admits no
$\mathrm{EF1}+\mathrm{PO}$ allocation.
\end{restatable}

Both agents in our construction have \emph{unweighted coverage valuations},
a particularly structured strict subclass of monotone submodular valuations.
Moreover, the same instance admits no $\mathrm{EF1}+\mathrm{PO}$ allocation
even when the items are interpreted as chores. The two valuations in our
instance differ only by a permutation of two items, and are identical
otherwise.
By contrast, an \efonepo\ allocation always exists for identical monotone
submodular valuations, for any number of agents. Thus, our example is a
minimal one, and shows that nonexistence arises even under a very limited
departure from identical valuations.

For goods, such a two-agent counterexample requires both valuations to be nonadditive: \efonepo\ always exists with one monotone submodular valuation and one additive valuation. In fact, we prove the following more general result.

\begin{restatable}{theorem}{EFonePOMixed}
\label{thm:EF1PO-mixed}
Every instance with one agent with a monotone
submodular valuation and $n-1$ agents with identical additive valuations
over goods admits an \efonepo\ allocation.
\end{restatable}

Our nonexistence result naturally motivates relaxations of \efonepo.
The result of \cite{CaragiannisKMPSW19} guarantees the existence of a
$\frac{1}{2}$-\efonepo\ allocation for goods. A natural question is how large this guaranteed approximation factor can be. For two agents, we show that the optimal factor lies in $\left[\frac{2}{3},\frac{3}{4}\right]$.

\begin{restatable}{theorem}{ApproxEFonePO}
\label{thm:alpha-ef1PO}
Every instance with two agents with monotone
submodular valuations over goods admits a
$\frac{2}{3}$-\efonepo\ allocation.
For every $\alpha>\frac{3}{4}$, there exists an instance with two agents
and seven goods with unweighted coverage valuations that admits no
$\alpha$-\efonepo\ allocation.
\end{restatable}

The proofs are given in \fullref{sec:alpha-efone-po}.

For chores, however, no positive approximation guarantee is possible even
with one coverage cost function and one additive cost function.
A simple instance with $k+1$ chores shows that, for every $\alpha>0$
and every integer $k\ge1$, an $\alpha$-\efkpo\ allocation need not exist
for two such agents.

We also consider the natural relaxation of \efkpo\ and determine
exact item thresholds for several valuation classes with two agents.

\begin{restatable}{theorem}{EFkPOThresholds}
\label{thm:efkpo}
For every integer $k\ge1$, the following statements hold
for instances with two agents:
\begin{enumerate}
\item\label{itm:submodular}
With monotone submodular valuations, an \efkpo\ allocation always
exists when there are at most $4k+1$ goods, whereas an instance
with $4k+2$ goods need not admit an \efkpo\ allocation.
\item\label{itm:subadditive}
With monotone subadditive valuations, an \efkpo\ allocation always
exists when there are at most $2k+1$ goods, whereas an instance
with $2k+2$ goods and identical monotone subadditive---in fact,
\textsc{XOS}---valuations need not admit an \efkpo\ allocation.
\item\label{itm:subadditive-additive}
With one monotone subadditive valuation and one additive valuation,
an \efkpo\ allocation always exists when there are at most $3k+1$ goods.
\end{enumerate}
\end{restatable}

All three thresholds in \fullref{thm:efkpo} are tight.
For the third threshold, a simple family of instances with $3k+2$ goods,
one monotone $\mathrm{XOS}$ valuation, and one additive valuation admits
no \efkpo\ allocation.
In particular, setting $k=1$ shows that our six-good instance in
\fullref{thm:EF1PO-nonexistence} is minimal in terms of number of goods,
as is the four-good example of \cite{CaragiannisKMPSW19} for monotone
subadditive valuations. The first two statements are proved in
\fullref{sec:efkpo}, and the third in \fullref{sec:one-subadditive}.

Beyond existence, we establish the following computational hardness result,
proved in \fullref{sec:hardness-goods}.

\begin{restatable}{theorem}{EFkPOHardness}
\label{thm:efonepo-NPhard}
Determining whether an instance with unweighted coverage valuations
over goods admits an \efkpo\ allocation, where $k\ge1$ is part of
the input, is \NP-hard, even with only three agents, one of whom
has an additive valuation. Hardness holds even for every fixed
integer $k\ge1$.
\end{restatable}

An \NP-hardness result for chores under coverage valuations is known
from \cite{HosseiniNW25}.

\subsection{Concurrent and Subsequent Work}

Independently, \cite{MackenzieS26} gave a two-agent, eight-good submodular instance that does not admit an \efonepo\ allocation. They also prove incompatibility with fractional Pareto optimality for weighted matroid-rank valuations, show compatibility for common-weight matroid rank valuations, and study $\mathrm{EF1}$ with approximate Pareto efficiency. Our example in \fullref{thm:EF1PO-nonexistence} uses unweighted coverage valuations and is minimal in the number of goods; moreover, the two valuations differ only by a permutation of two items.

Our instance in \fullref{thm:EF1PO-nonexistence} has zero marginal values for some goods with respect to some bundles.
A follow-up work \cite{Teh26} gave an instance with strictly positive marginal values that does not admit an $\mathrm{EF1}+\mathrm{PO}$ allocation. The paper gives an eight-good counterexample with strictly increasing submodular valuations and proves that eight goods are necessary: every two-agent instance with at most seven goods and strictly increasing valuations admits an $\mathrm{EF1}$+$\mathrm{PO}$ allocation without requiring submodularity. \fullref{thm:efkpo} in our paper thus highlights the contrast between zero and non-zero marginal values. \cite{Teh26} also contains an \NP-hardness reduction for three agents where all zero marginals are confined to eight fixed agent-good pairs.

In more recent work, \cite{MahabadiRT26} showed that envy-freeness up to any good ($\mathrm{EFX}$) and $\mathrm{PO}$ can be incompatible even for strictly positive additive valuations, using four agents and eight goods. For every $\alpha > \frac{2}{3}$, they also give a four-agent, nine-good instance admitting no $\alpha$-$\mathrm{EFX}$+$\mathrm{PO}$ allocation.

\subsection{Further Related Work}

The area of fair division of indivisible items has seen exceptional progress in recent years. We refer the reader to the recent survey ~\cite{AmanatidisABFLMVW23}. We now review some more of the existing work that discusses achieving simultaneous fairness and efficiency beyond additive valuations.

For matroid-rank valuations, equivalently submodular valuations with binary marginal gains, \cite{BenabbouCIZ21} proved that an $\mathrm{EF1}$ allocation maximizing utilitarian welfare always exists and is polynomial-time computable. Independently, \cite{BabaioffEF21} gave a truthful, Lorenz-dominating mechanism. \cite{BarmanV22} strengthened its strategic guarantee to group strategy-proofness; these outcomes are $\mathrm{EF1}$ and $\mathrm{PO}$. \cite{ViswanathanZ23} subsequently gave a general efficient framework for optimizing several fairness and welfare objectives in this domain.

\cite{KulkarniKM23} proved $\mathrm{EF1}$ together with maximum social welfare for identical assignment ($\mathrm{OXS}$) valuations, a strict subclass of submodular valuations with nonbinary marginals. \cite{PlautR20} showed the existence of the stronger notion of \emph{envy-freeness up to any good} ($\mathrm{EFX}$) and $\mathrm{PO}$ for identical monotone valuations with strictly positive marginal utilities. This result, together with Lemma 3.5 from \cite{BenabbouCIZ21}, implies $\mathrm{EF1}+\mathrm{PO}$ allocations for identical monotone submodular valuations (see \fullref{sec:identical}).

In recent work, \cite{CooksonSV26} proved that additive valuations admit balanced $\mathrm{EF1}$ allocations that are fractionally Pareto optimal among balanced fractional allocations. Balancedness requires bundle sizes to differ by at most one. Their proof uses the Knaster--Kuratowski--Mazurkiewicz lemma and extends to category constraints, yielding constrained $\mathrm{fPO}$ with envy eliminated by one deletion per category.

In more recent work, \cite{FeigeF26} combined $\mathrm{EF1}$ with a constant-factor maximin-share guarantee for submodular valuations, without Pareto optimality.

For additive mixed manna with unequal entitlements, \cite{TehMixed26} and \cite{LinLTZ26} give polynomial-time algorithms for $\mathrm{WEF1}$. \cite{LinLTZ26} also establish $\mathrm{fPO}$ together with weighted envy-freeness up to one transfer, allowing envy to be eliminated by moving one item between the compared bundles.

For binary additive chores, \cite{LinWZ26} compute allocations that are weighted envy-free up to any chore and fractionally Pareto optimal in polynomial time.


\subsection{Organization}

After the preliminaries in \fullref{sec:prelims},
\fullref{sec:nonexistence} presents our six-good counterexample.
\fullref{sec:additive-agents} studies instances with additive agents,
including the tight threshold for one subadditive and one additive
valuation. \fullref{sec:alpha-efone-po} proves the
$\frac{2}{3}$ existence and $\frac{3}{4}$ nonexistence bounds,
and \fullref{sec:efkpo} establishes the remaining item thresholds.
\fullref{sec:hardness-goods} proves computational hardness for goods.
The appendices give the existence proof for identical submodular
valuations in \fullref{sec:identical}, the six-item chore
counterexample in \fullref{sec:chores}, and the chore hardness
reduction in \fullref{sec:NP-hardness-chores}.

\section{Preliminaries} \label{sec:prelims}

Let $N = \{1,\ldots,n\}$ denote the set of agents and let $M$ be a finite set of indivisible goods. Each agent $i \in N$ has a valuation $v_i \colon 2^M \to \mathbb{R}_{\geq 0}$. We assume throughout that valuations are \emph{normalized}, meaning $v_i(\varnothing) = 0$, and \emph{monotone}, so $v_i(S) \leq v_i(T)$ whenever $S \subseteq T \subseteq M$. For $S \subseteq M$ and $g \in M \setminus S$, write $v_i(g \mid S) \coloneq v_i(S \cup \{g\}) - v_i(S)$ for the \emph{marginal value of $g$ given $S$}. We also write $v_i(g)$ for $v_i(\{g\})$.

A valuation $v_i$ is \emph{submodular} if it has diminishing marginal returns: for all $S \subseteq T \subseteq M$ and $g \in M \setminus T$, $v_i(g \mid S) \geq v_i(g \mid T)$. Equivalently, $v_i(S) + v_i(T) \geq v_i(S \cup T) + v_i(S \cap T)$ for all $S,T \subseteq M$. A valuation is \emph{subadditive} if $v_i(S \cup T) \leq v_i(S) + v_i(T)$ for all $S,T \subseteq M$; every submodular valuation under our assumptions is subadditive. \emph{Additive valuations}, given by $v_i(S) = \sum_{g \in S} v_i(\{g\})$, form a subclass of monotone submodular valuations. A \emph{matroid-rank valuation} is a monotone submodular valuation whose marginal values belong to $\{0,1\}$.

\emph{Weighted coverage valuation:} Suppose we are given a finite universe of coverage elements $U$, with weights $w_e \geq 0$ for each $e \in U$, and that each item $j \in M$ corresponds to a set
$A_j \subseteq U$. We say that {\em $j$ `covers' elements in $A_j$.} A valuation $v_i$ is a \emph{(weighted) coverage valuation} with respect to $U$, weights $\{w_e\}_{e \in U}$, and coverage description $\{A_j\}_{j \in M}$ if, for every $S \subseteq M$,
\[
v_i(S) = \sum_{e \in \bigcup_{j \in S} A_j} w_e.
\]
Equivalently, each $e \in U$ is a \emph{coverage atom} of weight $w_e$, contributing $w_e$ precisely when $S$ contains at least one item that covers $e$. Weighted coverage valuations form a strict subclass of monotone submodular valuations. If $w_e=1$ for every $e \in U$, the valuation is an \emph{unweighted coverage valuation}. Different agents may have different universes, weights, and coverage descriptions. In our hardness results, the coverage universes and covered sets are listed explicitly.

We recall the definition of $\mathrm{XOS}$ valuations, as given in \cite{LLN06}:

\emph{$\mathrm{XOS}$ valuation:} A valuation $v$ is $\mathrm{XOS}$ if there are nonnegative additive valuations $a_1,\ldots,a_t$ for some $t$ such that $v(S)=\max_{\ell\in [t]} a_\ell(S)$ for each $S\subseteq M$. Every monotone submodular valuation is $\mathrm{XOS}$, and every $\mathrm{XOS}$ valuation is subadditive.

An \emph{allocation} is a partition $X = (X_1,\ldots,X_n)$ of the set $M$ of goods among the agents $N$, where $X_i$ is the \emph{bundle} that is allocated to agent $i$. Agent $i$ envies agent $j$ in the allocation $X$ if $v_i(X_i) < v_i(X_j)$. The allocation $X$ is said to be \emph{envy-free up to one good} ($\mathrm{EF1}$) if, for every pair of distinct agents $i,j \in N$, either $X_j = \varnothing$, or there exists a good $g \in X_j$ such that $v_i(X_i) \geq v_i(X_j \setminus \{g\})$. By a slight abuse of notation, we use $v_i(X_j \setminus g)$ to denote $v_i(X_j \setminus \{g\})$.

More generally, for $\alpha > 0$ and an integer $k\geq 1$, an allocation of goods is $\alpha$-$\mathrm{EF}k$ if, for every pair of distinct agents $i,j \in N$, there exists $D\subseteq X_j$ with $\abs{D}\leq k$ such that $v_i(X_i)\geq \alpha v_i(X_j\setminus D)$. The case $\alpha = 1$ is called $\mathrm{EF}k$, and the case $k = 1$ is called $\alpha$-$\mathrm{EF1}$.

An allocation $Y$ \emph{Pareto dominates} an allocation $X$ if $v_i(Y_i) \geq v_i(X_i)$ for every $i \in N$, with strict inequality for at least one agent. An allocation is \emph{Pareto optimal} ($\mathrm{PO}$) if it is not Pareto dominated by any allocation.

Finally, we define the analogous notions for chores.
Here, $M$ denotes a set of chores, and $v_i : 2^M \to \mathbb{R}_{\geq 0}$ denotes agent $i$'s normalized monotone \emph{cost function}, so smaller values are preferred. The definitions of additivity, submodularity, and coverage apply to cost functions in the same way.

An allocation $A = (A_1,\ldots,A_n)$ of chores is $\mathrm{EF1}$ if, for every pair of distinct agents $i,j \in N$, either $A_i = \varnothing$, or there exists a chore $c \in A_i$ such that $v_i(A_i \setminus c) \leq v_i(A_j)$. It is $\mathrm{PO}$ if there is no allocation $B = (B_1,\ldots,B_n)$ such that $v_i(B_i) \leq v_i(A_i)$ for every $i \in N$, with strict inequality for at least one agent.

For $\alpha>0$ and an integer $k\geq 1$, an allocation of chores is $\alpha$-$\mathrm{EF}k$ if, for every pair of distinct agents $i,j \in N$, there exists $D\subseteq A_i$ with $\abs{D}\leq k$ such that $\alpha v_i(A_i\setminus D)\leq v_i(A_j)$. Again, $\alpha = 1$ gives $\mathrm{EF}k$, and $k = 1$ gives $\alpha$-$\mathrm{EF1}$. Thus deletions are from the envied bundle for goods and from the envious agent's own bundle for chores.

Thus, for goods as well as chores, $\alpha$-$\mathrm{EF}k$ is a relaxation of $\mathrm{EF}k$ for $\alpha\in (0,1)$.


\section{Nonexistence of \texorpdfstring{$\mathrm{EF1}+\mathrm{PO}$}{EF1+PO} Allocations for Submodular Valuations}\label{sec:nonexistence}

We construct an instance with two agents and six goods $M \coloneq \{p_1,p_2,p_3,q,r,r'\}$ that does not admit an \efonepo\ allocation, thereby proving \fullref{thm:EF1PO-nonexistence}. The valuations of both agents are defined via unweighted coverage over the universe $U \coloneq \{u_1,\ldots,u_7\}$. The valuations are shown in \fullref{tab:valuations}. The valuation of agent $1$ is denoted with a \textcolor{Green}{green} checkmark, and that of agent $2$ is denoted by an \textcolor{RedOrange}{orange} checkmark. The value of each set of items is the number of elements covered by them, counted without multiplicity.

\begin{table}[ht]
\centering
\[
\begin{array}{|c|c|c|c|c|c|c|c|}
\hline
\textrm{Items/elements} & u_1 & u_2 & u_3 & u_4 & u_5 & u_6 & u_7\\
\hline
\hline
p_1 & \makecell{\gch\\ \rch} & & & & \makecell{\gch\\ \rch} & & \\
\hline
p_2 & & \makecell{\gch\\ \rch} & & & & \makecell{\gch\\ \rch} & \\
\hline
p_3 & & & \makecell{\gch\\ \rch}  & & & & \makecell{\gch\\ \rch} \\
\hline
q & & & & & \makecell{\gch\\ \rch} & \makecell{\gch\\ \rch} & \makecell{\gch\\ \rch}\\
\hline
r & & & \rch & \makecell{\gch\\ \rch} & \gch& \gch & \gch \\
\hline
r' & & & \gch & \makecell{\gch\\ \rch} & \rch& \rch & \rch \\
\hline
\end{array}
\]
\caption{Summary of the valuations.}
\label{tab:valuations}
\end{table}

\EFonePONonexistence*

\begin{proof}
    The instance is described in \fullref{tab:valuations}. Observe that the maximum value of any bundle for either agent is $7$. For agent $1$, every bundle of value $7$ contains at least one of the four sets $\{p_1,p_2,p_3,r\}$, $\{p_1,p_2,p_3,r'\}$, $\{p_1,p_2,q,r'\}$, and $\{p_1,p_2,r,r'\}$. Their complements have values $4$, $5$, $3$, and $4$, respectively, for agent $2$.
    
    We now characterize the Pareto optimal allocations in this instance and show that none of them is $\mathrm{EF1}$.
    
    \begin{claim}\label{clm:val75}
      If either of the two agents gets a $7$-valued bundle, then the other agent can get a bundle of value at most $5$.  
    \end{claim}
    
    \begin{proof}
        For either of the two agents, items $p_1$ and $p_2$ are essential to make a $7$-valued bundle as $u_1$ and $u_2$ are covered only by $p_1$ and $p_2$ respectively. But the other agent, who does not get either of $p_1$ and $p_2$ can have a bundle of value at most $5$.
    \end{proof}
    
    \begin{claim}\label{clm:val66}
    There is no allocation that assigns a bundle of value $6$ to both agents.
    \end{claim}
    
    \begin{proof}
        Let, for the sake of contradiction, $X = (X_1,X_2)$ be an allocation such that $v_1(X_1)=v_2(X_2)=6$. Thus, as mentioned in the proof of \fullref{clm:val75}, $p_1$ and $p_2$ cannot be in the same bundle. Let $p_1\in X_1$. Then, $v_1(X_1)=6$ is possible if either $\{p_3,r\}\subseteq X_1$ or $\{r,r'\}\subseteq X_1$ or $\{q,r'\}\subseteq X_1$ holds. Hence, $X_2$ is contained in $\{p_2,q,r'\}$ or $\{p_2,p_3,q\}$ or $\{p_2,p_3,r\}$, respectively. These bundles have value at most $5$ as they do not cover $u_3$, $u_4$, and $u_5$ respectively, in addition to $u_1$. 

        A similar argument holds if $p_2\in X_1,p_1\in X_2$.
    \end{proof}

    The allocations $X = (\{p_1,p_2,p_3,r'\}, \{r,q\})$ and $Y = (\{q,r'\}, \{p_1,p_2,p_3,r\})$ have utility profiles $(7,5)$ and $(5,7)$ respectively. By \fullref{clm:val75} and \fullref{clm:val66}, every other utility profile is Pareto dominated by one of these two profiles. Among the four sets listed above, only $\{p_1,p_2,p_3,r'\}$ leaves agent $2$ value $5$, and assigning either remaining good to agent $1$ lowers agent $2$'s value. Thus $X$ is the only allocation with profile $(7,5)$. Since the valuations differ only by swapping $r$ and $r'$, the symmetric argument shows that $Y$ is the only allocation with profile $(5,7)$. Hence, $X$ and $Y$ are exactly the $\mathrm{PO}$ allocations.

    Neither allocation is $\mathrm{EF1}$. In $X$, agent $2$ has value $5$, whereas $v_2(X_1\setminus g) = 6$ for every $g\in X_1$. Likewise, in $Y$, agent $1$ has value $5$, whereas $v_1(Y_2\setminus g) = 6$ for every $g\in Y_2$.
\end{proof}

\begin{remark} \label{rem:identical}
    The instance that proves \fullref{thm:EF1PO-nonexistence} consists of two agents with nonidentical valuations. It can be shown that $\mathrm{EF1}+\mathrm{PO}$ allocations always exist for any number of agents with identical monotone submodular valuations. Indeed, every \emph{leximin} allocation \cite{PlautR20} is $\mathrm{EF1}+\mathrm{PO}$, by using a \emph{transferability lemma} for monotone submodular valuations that is due to \cite[Lemma~3.5]{BenabbouCIZ21}. We record a proof of this fact in \fullref{sec:identical}.
\end{remark}

The example in \fullref{thm:EF1PO-nonexistence} does not admit an \efonepo\ allocation when the items are chores, i.e., when the valuations are interpreted as disutilities. We show this in \fullref{sec:chores}.


\section{\texorpdfstring{$\mathrm{EF1}$+$\mathrm{PO}$}{EF1+PO} Allocations with Additive Agents}
\label{sec:additive-agents}

\fullref{thm:EF1PO-nonexistence} shows that $\mathrm{EF1}$+$\mathrm{PO}$ allocations may fail to exist even for two agents when they both have unweighted coverage valuations over goods. In this section, we show that an \efonepo\ allocation always exists when one agent has a monotone submodular valuation and the remaining agents have identical additive valuations. In particular, for two agents, one additive valuation suffices. We also show that replacing the submodular valuation by a subadditive valuation can invalidate this conclusion, even for two agents.

\subsection{One Monotone Submodular and Identical Additive Valuations}\label{sec:one-submodular}

Let $v_1$ denote agent $1$'s monotone submodular valuation and let $v$ be the common additive valuation of agents $2,\ldots,n$. We restate \fullref{thm:EF1PO-mixed}.

\EFonePOMixed*

\begin{proof}
    Among all allocations $(S,A_2,\ldots,A_n)$ satisfying $v(A_i)\geq v(S)$ for every $i\geq 2$, choose one that maximizes $v_1(S)$, subject to this minimizes $v(S)$, and subject to both minimizes $\sum_{i = 2}^n v(A_i)^2$.
    Such a choice exists: there are finitely many allocations, and the constraint holds whenever $S = \varnothing$.

    Fix a nonempty $A_j$, choose $g\in A_j$ maximizing $v(g)$, and let $q = v(A_j\setminus g)$. We claim that $q\leq v(A_i)$ for every other additive agent $i$. Otherwise, $q > v(A_i) \geq v(S)$, and $v(g) > 0$. Moving $g$ from $A_j$ to $A_i$ preserves the constraint and leaves $S$ unchanged, but decreases the sum of squares by $2 v(g) (q - v(A_i)) > 0$, a contradiction. Thus the additive agents are $\mathrm{EF1}$ towards one another. They do not envy agent $1$ by the constraint.

    We next check agent $1$'s comparison with $A_j$. If $q\leq v(S) + v(g)$, replace $S$ by $A_j\setminus \{g\}$ and replace $A_j$ by $S\cup \{g\}$. The new allocation satisfies the constraint: agent $j$ receives value at least $q$, and every other additive agent already has value at least $q$. Maximality therefore gives $v_1(S)\geq v_1(A_j\setminus g)$.

    Otherwise, $q > v(S) + v(g)$. Transferring any $h\in A_j$ to $S$ preserves the constraint: agent $j$ retains value at least $q$, while $v(S\cup h)\leq v(S) + v(g) < q$, and every other additive agent has value at least $q$. Maximality and monotonicity imply $v_1(S\cup h) = v_1(S)$ for every $h\in A_j$. By submodularity, these goods still add zero value when added one after another. Hence $v_1(S\cup A_j) = v_1(S)$, and therefore $v_1(A_j)\leq v_1(S)$. This proves $\mathrm{EF1}$.

    Finally, suppose $(S',A_2',\ldots,A_n')$ Pareto improves the chosen allocation. Additivity gives
    \begin{equation}
        v(S') = v(M) - \sum\limits_{i=2}^n v(A_i')\leq v(M) - \sum\limits_{i=2}^n v(A_i) = v(S). \label{eqn:our-allocation-PO}
    \end{equation}
    
    Since $v(A_i')\geq v(A_i)\geq v(S)\geq v(S')$, the new allocation also satisfies the constraint. Maximality forces $v_1(S') = v_1(S)$, so some additive agent must strictly improve. The inequality in \eqref{eqn:our-allocation-PO} is then strict, giving $v(S') < v(S)$ and contradicting the second choice (minimizing $v(S)$). Thus the allocation is Pareto optimal.
\end{proof}

\begin{remark}
    The conclusion fails for \emph{chores}, even with one coverage cost function and one additive cost function. Indeed, fix $\alpha > 0$ and $k\geq 1$. Define an instance with $k + 1$ chores, with costs $v_1(S) = \min\{1,\abs{S}\}$ and $v_2(S) = \abs{S}$. The first cost function is monotone submodular, and in fact a coverage function: every chore covers the same single universe element. The second one is clearly additive.

    If both agents receive a nonempty bundle, giving all chores to agent $1$ leaves its cost unchanged and strictly decreases agent $2$'s cost. Thus every $\mathrm{PO}$ allocation gives all chores to one agent. After removing at most $k$ chores from that agent's bundle, its cost remains positive, whereas its cost for the other agent's empty bundle is zero. Hence no $\mathrm{PO}$ allocation is $\alpha$-$\mathrm{EF}k$, for any $\alpha > 0$. The number of chores in this negative instance is tight, as $\mathrm{EF}k$+$\mathrm{PO}$ trivially exists when there are at most $k$ chores in the instance.
\end{remark}

\subsection{One Monotone Subadditive and One Additive Valuation}\label{sec:one-subadditive}

Submodularity is essential even for two agents. The following lemma gives a $3k+2$-good counterexample with one monotone XOS valuation and one additive valuation. Proposition~\ref{prop:one-subadditive} proves existence with at most $3k+1$ goods, establishing the third statement of Theorem~\ref{thm:efkpo}.

\begin{lemma}\label{lem:one-subadditive}
    For every $k\geq 1$, there is an instance with two agents and $3k+2$ goods where one agent has a monotone \emph{subadditive} valuation and the other agent has an additive valuation, which does not admit an $\mathrm{EF}k$+$\mathrm{PO}$ allocation. 
\end{lemma}
\begin{proof}
    The construction of the instance is inspired by the instance in \cite{CaragiannisKMPSW19}, which consists of a special good that they call $c$ (which we will call $a$ here) and there are three goods $\{a,b,d\}$ which are interchangeable in their instance that we will denote here by a set $B$. Let $M = \{a\} \mathbin{\dot\cup} B$, where $\abs{B} = 3k+1$, and define
    \[
    v_1(S) = \max\{2k\cdot  \mathbf{1}[a\in S], \abs{S\cap B}\}, \quad v_2(S) = \mathbf{1}[a\in S] + 2\abs{S\cap B}.
    \]
    In fact, $v_1$ is monotone $\mathrm{XOS}$, since it is the maximum of two nonnegative additive valuations.

    In a $\mathrm{PO}$ allocation, agent $1$ receives either nothing, only $a$, or at least $2k+1$ goods of $B$ without $a$. Indeed, if its bundle contains both $a$ and goods of $B$, one of these parts can be returned to agent $2$ without decreasing $v_1$. If it receives between one and $2k$ goods of $B$ without $a$, exchanging them for $a$ weakly improves agent $1$ and strictly improves agent $2$.

    In the first two cases, agent $1$ has value at most $2k$, but values agent $2$'s bundle at least $2k+1$ after any $k$ deletions. In the remaining case, agent $2$ has value at most $2k+1$, but values agent $1$'s bundle at least $2k+2$ after any $k$ deletions. Thus no $\mathrm{PO}$ allocation is $\mathrm{EF}k$.
\end{proof}
    
    The number of goods in this instance is tight as shown below.

    \begin{proposition}\label{prop:one-subadditive}
        Every instance with two agents and at most $3k+1$ goods, where $v_1$ is a monotone subadditive valuation and $v_2$ is an additive valuation, admits an $\mathrm{EF}k$+$\mathrm{PO}$ allocation.
    \end{proposition}

    \begin{proof}
        Suppose an instance admits no $\mathrm{EF}k$+$\mathrm{PO}$ allocation. We show that it has at least $3k+2$ goods.

        List all $\mathrm{PO}$ allocations in nondecreasing order of agent $1$'s value, and hence in nonincreasing order of agent $2$'s value. Exactly one agent $\mathrm{EF}k$-envies the other in each allocation: otherwise the allocation is $\mathrm{EF}k$, or exchanging the bundles improves both agents. The first allocation has agent $1$ envying agent $2$, and the last has the reverse direction. Thus there are consecutive $\mathrm{PO}$ allocations $A = (A_1,A_2)$ and $B = (B_1,B_2)$ such that agent $1$ $\mathrm{EF}k$-envies agent $2$ in $A$, while agent $2$ $\mathrm{EF}k$-envies agent $1$ in $B$. Let $a_1 = v_1(A_1)$ and $b_2 = v_2(B_2)$.

        No allocation $X$ can satisfy both $v_1(X_1) > a_1$ and $v_2(X_2) > b_2$; a $\mathrm{PO}$ allocation weakly dominating $X$ would lie strictly between $A$ and $B$, contradicting their consecutiveness.

        Let $P = A_1\cap B_1$, $Q = A_1\cap B_2$, $R = A_2\cap B_1$, and $S = A_2\cap B_2$. Thus $A = (P\cup Q, R\cup S)$ and $B = (P\cup R, Q\cup S)$. Agent $1$'s envy in $A$ and the preceding observation imply
        \begin{equation}
            v_2(P\cup Q\cup D)\leq b_2, \text{ for every $D\subseteq R\cup S$ with $\abs{D}\leq k$}. \label{eqn:main-claim-reduced-bundle}
        \end{equation}
        
        First, $\abs{R}\geq 2k+1$. Otherwise partition $R$ into $G$, $H$ with $\abs{G},\abs{H}\leq k$. The allocation $(Q\cup S\cup H, P\cup G)$ satisfies $v_1(Q\cup S\cup H)\geq v_1(S\cup H) > a_1$ by agent $1$'s envy in $A$, and $v_2(P\cup G) > b_2$ by agent $2$'s envy in $B$, a contradiction.

        Next, $\abs{S}\geq k$. Otherwise, taking $D = S\cup \{r\}$ in \eqref{eqn:main-claim-reduced-bundle} gives $v_2(P) + v_2(r)\leq 0$ for every $r\in R$. Nonnegativity and additivity imply $v_2(P\cup R) = 0$, contradicting agent $2$'s envy in $B$.

        Finally, $P\cup Q \neq \varnothing$. Otherwise $a_1 = 0$. Agent $1$'s envy gives $v_1(M) > 0$, so subadditivity gives a good $g$ with $v_1(g) > 0$. Since $k\geq 1$, agent $2$'s envy in $B$ gives $v_2(M\setminus g)\geq v_2(B_1\setminus g) > b_2$. The allocation $(\{g\}, M\setminus\{g\})$ is again a contradiction.

        Therefore, $\abs{M} = \abs{P\cup Q} + \abs{R} + \abs{S} \geq 1 + (2k+1) + k = 3k+2$.
    \end{proof}
    

\section{\texorpdfstring{$\alpha$-$\mathrm{EF1}$+$\mathrm{PO}$}{alpha-EF1+PO} Allocations}
\label{sec:alpha-efone-po}

We showed that $\mathrm{EF1}+\mathrm{PO}$ allocations may fail to exist
when agents have submodular valuations, whereas existence of
$\mathrm{MEF1}+\mathrm{PO}$ allocations is guaranteed by
\cite{CaragiannisKMPSW19}. Since $\mathrm{MEF1}$ implies
$\frac{1}{2}$-$\mathrm{EF1}$, a natural question is how large $\alpha$
can be while guaranteeing an $\alpha$-\efonepo\ allocation.
In the example in \fullref{tab:valuations}, every $\mathrm{PO}$
allocation is exactly $\frac{5}{6}$-$\mathrm{EF1}$, ruling out a
guarantee strictly greater than $\frac{5}{6}$.
For two agents, we strengthen the nonexistence bound to $\frac{3}{4}$
and improve the existence guarantee from $\frac{1}{2}$ to $\frac{2}{3}$.
The existence proof is given in \fullref{sec:2by3-EF1PO}.

\ApproxEFonePO*

\subsection{Nonexistence of \texorpdfstring{$\left(\frac{3}{4}+\varepsilon\right)$-$\mathrm{EF1}+\mathrm{PO}$}{(3/4 + epsilon)-EF1+PO} Allocations for Goods}

We prove the nonexistence statement of \fullref{thm:alpha-ef1PO}.
Fix $\alpha>\frac{3}{4}$.

\begin{proof}[{Proof of the nonexistence statement in \fullref{thm:alpha-ef1PO}}]
Choose a positive integer $t$ large enough that
$3t+4<4\alpha t$ and $2t+4<3\alpha t$.
The instance has two agents and seven goods
$M=P\cup\{q,r,r'\}$, where $P=\{p_1,p_2,p_3,p_4\}$.

Each agent has an unweighted coverage valuation over a universe
of size $5t+4$, consisting of five pairwise disjoint sets
$U_1,\ldots,U_5$, each containing $t$ elements, and four further
elements $w_1,w_2,w_3,w_4$. Each agent values a bundle by the
number of distinct elements covered by its goods.
The coverage sets are
\[
\begin{array}{|c|c|c|}
\hline
\text{Good}
 & \text{Elements covered for agent 1}
 & \text{Elements covered for agent 2}\\
\hline
p_1 & U_1\cup\{w_2\} & U_1\cup\{w_2\}\\
p_2 & U_2\cup\{w_3\} & U_2\cup\{w_3\}\\
p_3 & U_3\cup\{w_4\} & U_3\\
p_4 & U_4 & U_4\cup\{w_4\}\\
q & U_4\cup U_5\cup\{w_2,w_3,w_4\}
  & U_3\cup U_5\cup\{w_2,w_3,w_4\}\\
r & U_5 & U_4\cup\{w_1\}\\
r' & U_3\cup\{w_1\}
   & U_5\cup\{w_1,w_2,w_3,w_4\}\\
\hline
\end{array}
\]
The allocations
\[
 A=(\{q,r'\},P\cup\{r\}),\qquad
 B=(P\cup\{r'\},\{q,r\})
\]
have utility pairs $(3t+4,4t+4)$ and $(4t+4,3t+4)$,
respectively.

For an allocation $X=(X_1,X_2)$, let $n_i$ be the number of
sets among $U_1,\ldots,U_5$ covered by $X_i$ for agent $i$.
We first show that
\begin{equation}\label{eqn:cover-at-most-7}
                         n_1+n_2\le7.
\end{equation}
The sets $U_1,U_2$ contribute exactly two to this sum, since
only $p_1,p_2$ cover them. Suppose both agents covered all
of $U_3,U_4,U_5$. Agent 1 would need $p_3$ or $r'$ to cover
$U_3$. Agent 2 would then need $q$ to cover $U_3$ or $U_5$,
respectively. This forces $p_4$ and $r$ to agent 1 to cover
$U_4$ and $U_5$, leaving agent 2 unable to cover $U_4$.
Thus the remaining three sets contribute at most five,
proving \eqref{eqn:cover-at-most-7}.

We next show that a bundle covering all five sets for one
agent still covers at least three sets for the other agent
after any one good is deleted. Such a bundle $T$ contains
$p_1,p_2$ and at least two other goods, since no single good
covers all of $U_3,U_4,U_5$.

If the deleted good is not $p_1$ or $p_2$, these two goods
remain, together with at least one other good. For the other
agent, they cover $U_1,U_2$ and at least one of $U_3,U_4,U_5$.

If the deleted good is $p_1$ or $p_2$, the other good in this
pair remains. It suffices to find two covered sets among
$U_3,U_4,U_5$. If $q\in T$, then $q$ remains and covers two
such sets. Otherwise, $T$ contains at least three of
$p_3,p_4,r,r'$. Any three of these cover at least two sets
for either agent: among these four goods, the only pair
covering the same set is $\{p_3,r'\}$ for agent 1 and
$\{p_4,r\}$ for agent 2. Thus at least three sets remain
covered in every case.

Any allocation with $n_i=5$ therefore fails
$\alpha$-$\mathrm{EF1}$. By \eqref{eqn:cover-at-most-7},
the other agent $j$ has $n_j\le2$, and hence
\[
 v_j(X_j)\le2t+4<3\alpha t
       \le\alpha v_j(X_i\setminus\{g\})
       \qquad(g\in X_i).
\]

It remains to consider $\mathrm{PO}$ allocations with
$n_1,n_2\le4$. If $n_1\le3$, allocation $A$ weakly improves
both utilities. If $n_1=4$, then $n_2\le3$ by
\eqref{eqn:cover-at-most-7}, and allocation $B$ weakly
improves both utilities. Pareto optimality therefore forces
the utility pair to equal that of $A$ or $B$.
In particular, both agents cover all four elements
$w_1,w_2,w_3,w_4$.

Covering $w_1$ forces $r'\in X_1$ and $r\in X_2$.
The agent without $q$ must receive $p_1,p_2$ to cover
$w_2,w_3$, and must receive $p_3$ if this is agent 1,
or $p_4$ if this is agent 2, to cover $w_4$.
If the remaining $p$-good were held by the agent with $q$,
moving it to the other agent would leave its owner's value
unchanged and strictly increase the other agent's value.
Thus all of $P$ must go to the agent without $q$.
Hence $X=A$ or $X=B$.

Finally, neither $A$ nor $B$ is $\alpha$-$\mathrm{EF1}$.
The agent with smaller utility has value $3t+4$.
According to that agent, the other bundle contains five
goods covering five different sets $U_j$, so deleting any
one good leaves value at least $4t$.
Since $3t+4<4\alpha t$, this agent fails
$\alpha$-$\mathrm{EF1}$. Therefore no
$\alpha$-$\mathrm{EF1}+\mathrm{PO}$ allocation exists.
\end{proof}

\paragraph{Nonexistence for Chores.}
As a separate remark, for every $\alpha>0$ and integer $k\ge1$,
we give an instance with two identical unweighted coverage cost
functions that admits no $\alpha$-\efkpo\ allocation.

\begin{proof}[Proof of the chore remark]
Consider two agents with identical unweighted coverage costs
over $k+1$ chores $M=\{c_1,\ldots,c_{k+1}\}$, all of which
cover the same universe element $u$. Thus $v(\varnothing)=0$
and $v(S)=1$ for every nonempty $S\subseteq M$.

The only $\mathrm{PO}$ allocations are
$X=(M,\varnothing)$ and $Y=(\varnothing,M)$, since every
allocation in which both agents receive a nonempty bundle
has cost profile $(1,1)$ and is Pareto dominated by both
$X$ and $Y$. For every $D\subseteq M$ with $|D|\le k$,
the bundle $M\setminus D$ is nonempty and hence has cost $1$,
whereas the other agent's empty bundle has cost $0$.
Thus neither $X$ nor $Y$ is $\alpha$-$\mathrm{EF}k$ for any
$\alpha>0$, so no $\alpha$-\efkpo\ allocation exists.
The number $k+1$ of chores is tight: with at most $k$ chores,
every allocation is $\mathrm{EF}k$, and hence any
$\mathrm{PO}$ allocation is $\mathrm{EF}k+\mathrm{PO}$.
\end{proof}

\subsection{Existence of \texorpdfstring{$2/3\text{-}\mathrm{EF1}\text{+}\mathrm{PO}$}{2/3-EF1+PO} for Two Agents}
\label{sec:2by3-EF1PO}

We prove the existence statement of \fullref{thm:alpha-ef1PO}.

\begin{proof}[{Proof of the existence statement in \fullref{thm:alpha-ef1PO}}]
Suppose, for contradiction, that no such allocation exists.
List all Pareto-optimal ($\mathrm{PO}$) allocations in nondecreasing
order of agent $1$'s value and, hence, in nonincreasing order of
agent $2$'s value.

In each allocation, exactly one agent fails
$\frac{2}{3}$-$\mathrm{EF1}$: at least one fails by assumption,
and if both failed, exchanging their bundles would strictly
improve both agents' values. In the first allocation, agent $2$
has value $v_2(M)$, whereas in the last, agent $1$ has value
$v_1(M)$. Thus there are consecutive $\mathrm{PO}$ allocations
\[
 A=(A_1,A_2),\qquad B=(B_1,B_2)
\]
such that agent $1$ fails $\frac{2}{3}$-$\mathrm{EF1}$ in $A$
and agent $2$ fails it in $B$. Write $a_1=v_1(A_1)$ and
$b_2=v_2(B_2)$. The two failures imply
\begin{equation}\label{eqn:two-thirds-ef1-failure}
 v_1(A_2\setminus\{g\})>\frac{3}{2}a_1,\qquad
 v_2(B_1\setminus\{g\})>\frac{3}{2}b_2
 \quad(g\in M).
\end{equation}
Here, removing a good outside a bundle leaves it unchanged.

\begin{claim}\label{clm:consecutive-po}
There is no allocation $X$ such that $v_1(X_1)>a_1$
and $v_2(X_2)>b_2$.
\end{claim}

\begin{proof}
A $\mathrm{PO}$ allocation weakly improving both values of
such an $X$ would lie strictly between $A$ and $B$:
Pareto optimality of $A$ forces its second value below
$v_2(A_2)$, and Pareto optimality of $B$ forces its first
value below $v_1(B_1)$. This contradicts their consecutiveness.
\end{proof}

Let
\[
 P=A_1\cap B_1,\quad Q=A_1\cap B_2,\quad
 R=A_2\cap B_1,\quad S=A_2\cap B_2.
\]
These sets are pairwise disjoint, and
\[
 A=(P\cup Q,R\cup S),\qquad B=(P\cup R,Q\cup S).
\]

\begin{claim}\label{clm:others-lower-bounds}
We have $v_2(P)>\frac{1}{2}b_2$ and
$v_1(S)>\frac{1}{2}a_1$.
\end{claim}

\begin{proof}
Since $v_1(A_2)>v_1(A_1)$, submodularity implies that some
$g\in A_2$ satisfies $v_1(A_1\cup\{g\})>a_1$.
Indeed, if every such good had zero marginal value at $A_1$,
adding all of $A_2$ would leave its value unchanged,
contradicting $v_1(A_2)>a_1$.

Applying \fullref{clm:consecutive-po} to
$(A_1\cup\{g\},A_2\setminus\{g\})$ gives
$v_2(A_2\setminus\{g\})\le b_2$.
Since $B_1\setminus\{g\}\subseteq
P\cup(A_2\setminus\{g\})$, subadditivity gives
\[
 \frac{3}{2}b_2
 <v_2(B_1\setminus\{g\})
 \le v_2(P)+v_2(A_2\setminus\{g\})
 \le v_2(P)+b_2.
\]
Hence $v_2(P)>\frac{1}{2}b_2$.
The other inequality follows by exchanging the agents
and the roles of $A$ and $B$.
\end{proof}

Since $v_2(B_1)>b_2$, applying \fullref{clm:consecutive-po}
to $(B_2,B_1)$ gives $v_1(S)\le v_1(B_2)\le a_1$.
On the other hand, $v_1(S\cup R)>\frac{3}{2}a_1$.
Choose an inclusion-minimal $G\subseteq R$ with
$v_1(S\cup G)>a_1$, and put $H=R\setminus G$.

The set $G$ is nonempty. For any $r\in G$, minimality gives
$v_1(S\cup(G\setminus\{r\}))\le a_1$.
Submodularity and \fullref{clm:others-lower-bounds} now give
\begin{align*}
 v_1(S\cup H)
 &\ge v_1(A_2\setminus\{r\})+v_1(S)
       -v_1(S\cup(G\setminus\{r\}))\\
 &>\frac{3}{2}a_1+\frac{1}{2}a_1-a_1=a_1.
\end{align*}
Thus both $S\cup G$ and $S\cup H$ have value greater than
$a_1$ to agent $1$.

Apply \fullref{clm:consecutive-po} to the allocations
\[
 (S\cup G,P\cup Q\cup H),\qquad
 (S\cup H,P\cup Q\cup G).
\]
Monotonicity gives $v_2(P\cup H)\le b_2$ and
$v_2(P\cup G)\le b_2$. But submodularity and
\fullref{clm:others-lower-bounds} imply
\begin{align*}
 2b_2
 &\ge v_2(P\cup G)+v_2(P\cup H)\\
 &\ge v_2(P\cup R)+v_2(P)
 >\frac{3}{2}b_2+\frac{1}{2}b_2=2b_2,
\end{align*}
a contradiction.
\end{proof}


\section{Existence of \texorpdfstring{$\mathrm{EF}k+\mathrm{PO}$}{EFk+PO} Allocations}
\label{sec:efkpo}

Throughout this section, $k\ge1$ is an integer.
We prove the first two statements of \fullref{thm:efkpo}.
The third statement was proved in
Proposition~\ref{prop:one-subadditive}.

\EFkPOThresholds*

\subsection{\texorpdfstring{$\mathrm{EF}k+\mathrm{PO}$}{EFk+PO} Allocations for Monotone Submodular Valuations}

We prove part~\ref{itm:submodular} of \fullref{thm:efkpo}.

\begin{lemma}\label{lem:efkpo-existence-submodular}
Every instance with two agents, at most $4k+1$ goods, and
monotone submodular valuations admits an \efkpo\ allocation.
\end{lemma}

\begin{proof}
Consider a two-agent instance with monotone submodular valuations
that admits no \efkpo\ allocation. We show that it must contain
at least $4k+2$ goods.

Let $M$ be the set of goods, and let $v_1,v_2$ be the valuations.
List all Pareto-optimal ($\mathrm{PO}$) allocations in nondecreasing
order of agent $1$'s value and, consequently, in nonincreasing
order of agent $2$'s value. In the first allocation, agent $2$
has value $v_2(M)$, and in the last, agent $1$ has value $v_1(M)$.

By assumption, none of these allocations is $\mathrm{EF}k$.
Exactly one agent $\mathrm{EF}k$-envies the other in each
allocation, since otherwise exchanging the bundles would
strictly improve both agents.
In the first allocation, agent $1$ must envy agent $2$,
whereas in the last, agent $2$ must envy agent $1$.
Thus there are consecutive $\mathrm{PO}$ allocations
$A=(A_1,A_2)$ and $B=(B_1,B_2)$ such that agent $1$
$\mathrm{EF}k$-envies agent $2$ in $A$, while agent $2$
$\mathrm{EF}k$-envies agent $1$ in $B$.%
\footnote{An agent $\mathrm{EF}k$-envies another agent if her
envy cannot be eliminated by removing at most $k$ goods
from the envied bundle.}
Let $a_1=v_1(A_1)$, $a_2=v_2(A_2)$, $b_1=v_1(B_1)$,
and $b_2=v_2(B_2)$. By the ordering, $a_1\le b_1$
and $a_2\ge b_2$.

\begin{claim}\label{clm:no-intermediate}
There is no allocation $X$, $\mathrm{PO}$ or otherwise, with
$v_1(X_1)>a_1$ and $v_2(X_2)>b_2$.
\end{claim}

\begin{proof}
Suppose such an allocation $X$ exists.
First assume that $X$ is $\mathrm{PO}$.
Pareto optimality of $A$ gives $v_2(X_2)<a_2$, and Pareto
optimality of $B$ gives $v_1(X_1)<b_1$.
Thus $X$ lies strictly between $A$ and $B$ in the ordering,
contradicting their consecutiveness.

If $X$ is not $\mathrm{PO}$, let $Y$ be a $\mathrm{PO}$
allocation Pareto dominating $X$. Then
$v_1(Y_1)>a_1$ and $v_2(Y_2)>b_2$, and the preceding argument
again gives a contradiction.
\end{proof}

Let $P=A_1\cap B_1$, $Q=A_1\cap B_2$,
$R=A_2\cap B_1$, and $S=A_2\cap B_2$.
These sets are pairwise disjoint, and
$A=(P\cup Q,R\cup S)$ and $B=(P\cup R,Q\cup S)$.

Agent $1$'s $\mathrm{EF}k$-envy in $A$ gives
$a_1<v_1((R\cup S)\setminus D)$ whenever
$D\subseteq R\cup S$ and $|D|\le k$.
Applying \fullref{clm:no-intermediate} to
$((R\cup S)\setminus D,P\cup Q\cup D)$ gives
\begin{equation}\label{eqn:v2}
 v_2(P\cup Q\cup D)\le b_2
 \quad\text{for every }D\subseteq R\cup S
 \text{ with }|D|\le k.
\end{equation}
Similarly, agent $2$'s $\mathrm{EF}k$-envy in $B$ gives
$b_2<v_2((P\cup R)\setminus E)$ whenever
$E\subseteq P\cup R$ and $|E|\le k$.
Applying \fullref{clm:no-intermediate} to
$(Q\cup S\cup E,(P\cup R)\setminus E)$ gives
\begin{equation}\label{eqn:v1}
 v_1(Q\cup S\cup E)\le a_1
 \quad\text{for every }E\subseteq P\cup R
 \text{ with }|E|\le k.
\end{equation}

We now bound $|R|$, $|P|$, and $|S|$.

\begin{lemma}\label{lem:sizeR}
We have $|R|\ge2k+1$.
\end{lemma}

\begin{proof}
Otherwise, partition $R$ into disjoint sets $G,H$ with
$|G|,|H|\le k$.
Agent $1$'s envy in $A$, with $D=G$, gives
$a_1<v_1(S\cup H)$.
But \eqref{eqn:v1}, with $E=H$, gives
$v_1(Q\cup S\cup H)\le a_1$, contradicting monotonicity.
\end{proof}

\begin{lemma}\label{lem:sizePS}
We have $|P|\ge k$ and $|S|\ge k$.
\end{lemma}

\begin{proof}
Suppose $|P|\le k-1$.
Setting $E=P$ in \eqref{eqn:v1} and using monotonicity gives
$v_1(P\cup Q\cup S)=a_1$.
For every $r\in R$, we may also set $E=P\cup\{r\}$,
obtaining $v_1(P\cup Q\cup S\cup\{r\})=a_1$.
Thus each $r\in R$ has zero marginal value at $P\cup Q\cup S$.
By submodularity, adding these goods one at a time gives
$v_1(M)=a_1$. This contradicts agent $1$'s envy in $A$.

The symmetric argument, using \eqref{eqn:v2} and agent $2$'s
envy in $B$, gives $|S|\ge k$.
\end{proof}

These bounds imply $|M|\ge4k+1$. It remains to rule out equality.
Suppose $|M|=4k+1$. Then $Q=\varnothing$,
$|P|=|S|=k$, and $|R|=2k+1$.
Setting $E=P$ in \eqref{eqn:v1} gives
$v_1(P\cup S)=v_1(P)=a_1$.
Similarly, setting $D=S$ in \eqref{eqn:v2} gives
$v_2(P\cup S)=v_2(S)=b_2$.

Let $R_1=\{r\in R:v_1(r\mid P)>0\}$ and
$R_2=\{r\in R:v_2(r\mid S)>0\}$.

\begin{claim}\label{clm:R1R2}
We have $|R_1|\ge k+1$, $|R_2|\ge k+1$, and
$R_1\cap R_2=\varnothing$.
\end{claim}

\begin{proof}
Suppose $|R_1|\le k$.
Each $r\in R\setminus R_1$ has zero marginal value at $P$,
and hence also at $P\cup S$, by submodularity and monotonicity.
Adding these goods one at a time gives
$v_1((R\setminus R_1)\cup P\cup S)=a_1$.
But agent $1$'s envy in $A$, with $D=R_1$, gives
$a_1<v_1(S\cup(R\setminus R_1))$, a contradiction.
The symmetric argument gives $|R_2|\ge k+1$.

Suppose now that $r\in R_1\cap R_2$.
Consider $X=(P\cup\{r\},(R\setminus\{r\})\cup S)$.
We have $v_1(X_1)>v_1(P)=a_1$.
Since $|R_2|\ge k+1\ge2$, there is some
$r'\in R_2\setminus\{r\}$, giving
$v_2(X_2)\ge v_2(S\cup\{r'\})>b_2$.
This contradicts \fullref{clm:no-intermediate}.
\end{proof}

Consequently, $|R|\ge|R_1|+|R_2|\ge2k+2$, contradicting
$|R|=2k+1$. Therefore $|M|\ge4k+2$, proving the lemma.
\end{proof}

The bound in \fullref{lem:efkpo-existence-submodular} is tight
even for unweighted coverage valuations.

\begin{lemma}\label{lem:efk-po-example}
For every $k\ge1$, there exists an instance with two agents,
$4k+2$ goods, and unweighted coverage valuations that admits
no \efkpo\ allocation.
\end{lemma}

\begin{proof}
Partition the goods as $M=P\cup\{q\}\cup R\cup R'$, where
$|P|=2k+1$, $R=\{r_1,\ldots,r_k\}$, and
$R'=\{r'_1,\ldots,r'_k\}$. Thus $|M|=4k+2$.
Let $\sigma$ fix $P\cup\{q\}$ and interchange $r_i,r'_i$
for each $i\in[k]$. Extend $\sigma$ elementwise to subsets of $M$.

Define agent $1$'s coverage universe as the disjoint union
$U=U_P\cup U_R\cup U_\Lambda$, where:
\begin{enumerate}
    \item $U_P$ consists of $k+1$ elements, each covered by
    every good in $P$;
    \item $U_R=\{u_1,\ldots,u_k\}$, where $u_i$ is covered
    only by $r_i$;
    \item
    $U_\Lambda=
    \{u_{g,T,s}:g\in R\cup\{q\},
    \ T\in\binom{P\cup R'}{k+1},\ s\in[k+2]\}$,
    where $u_{g,T,s}$ is covered precisely by the goods
    in $\{g\}\cup T$.
\end{enumerate}
Let $v_1$ be this coverage valuation and define
$v_2(S)=v_1(\sigma(S))$.
Both valuations are unweighted coverage valuations, and
$|U|=2k+1+(k+1)(k+2)\binom{3k+1}{k+1}$.

For $S\subseteq M$, define
$\ell(S)=|U|-v_1(M\setminus S)$, the loss incurred by agent $1$
when $S$ is not allotted to it.
Write $a=|S\cap P|$, $b=\mathbf{1}[q\in S]$,
$c=|S\cap R|$, and $d=|S\cap R'|$.
An element is lost exactly when all the goods covering it
belong to $S$. Therefore,
\begin{equation}\label{eqn:loss-function}
 \ell(S)=(k+1)\mathbf{1}[a=2k+1]+c
          +(k+2)(b+c)\binom{a+d}{k+1},
\end{equation}
where $\binom{x}{k+1}=0$ for $x\le k$.

Define $\Phi(S)=M\setminus\sigma(S)$.
The counts associated with $\Phi(S)$ are
$(2k+1-a,1-b,k-d,k-c)$, and the utility profile of
$(M\setminus S,S)$ is
$(|U|-\ell(S),|U|-\ell(\Phi(S)))$.

\begin{claim}\label{clm:po-structure}
The only Pareto-optimal allocations are
$A=(R\cup\{q\},P\cup R')$ and
$B=(P\cup R,\{q\}\cup R')$.
Their utility profiles are $(|U|-(k+1),|U|)$ and
$(|U|,|U|-(k+1))$, respectively.
\end{claim}

\begin{proof}
We first show that, for every $S\subseteq M$,
\begin{equation}\label{eqn:loss-dichotomy}
             \max\{\ell(S),\ell(\Phi(S))\}\ge k+1.
\end{equation}
Suppose both losses were smaller than $k+1$.
Since the final term in \eqref{eqn:loss-function} is either
zero or at least $k+2$, we obtain
\[
 a\le2k,\qquad b+c=0\ \text{or}\ a+d\le k.
\]
Applying the same argument to $\Phi(S)$ gives
\[
 a\ge1,\qquad
 (d=k\ \text{and}\ b=1)\ \text{or}\ a+c\ge2k+1.
\]
If $b+c=0$, then $b=c=0$. The first alternative in the second
display is impossible, and the other gives $a\ge2k+1$,
a contradiction.
Otherwise, $a+d\le k$.
The alternative $d=k,b=1$ forces $a=0$, again a contradiction.
The remaining alternative gives $a+c\ge2k+1$, although
$a\le k$ and $c\le k$. This proves \eqref{eqn:loss-dichotomy}.

We next establish the equality cases:
\begin{align}
 \ell(S)=k+1,\ \ell(\Phi(S))=0
 &\iff S=P\cup R',
 \label{eqn:first-PO-unique}\\
 \ell(S)=0,\ \ell(\Phi(S))=k+1
 &\iff S=\{q\}\cup R'.
 \label{eqn:second-PO-unique}
\end{align}
The reverse implication in \eqref{eqn:first-PO-unique}
follows from \eqref{eqn:loss-function}.
Conversely, suppose $\ell(S)=k+1$ and $\ell(\Phi(S))=0$.
The linear term in $\ell(\Phi(S))$ forces $d=k$, and its
first term forces $a\ge1$.
If $b=1$, the final term of $\ell(S)$ is at least $k+2$,
so $b=0$.
The final term of $\ell(\Phi(S))$ must vanish.
Its coefficient is now $(k-d)+(1-b)=1$, giving
$(2k+1-a)+(k-c)\le k$, or $a+c\ge2k+1$.
If $c>0$, the final term of $\ell(S)$ is again at least
$k+2$. Hence $c=0$ and $a=2k+1$, proving
$S=P\cup R'$.
Since $\Phi$ is an involution,
\eqref{eqn:second-PO-unique} follows by symmetry.

The allocations $A,B$ correspond to $S=P\cup R'$ and
$S=\{q\}\cup R'$, giving their stated profiles.
For any allocation $(M\setminus S,S)$,
$\ell(S)\ge k+1$ implies that $A$ weakly improves both
utilities, and $\ell(\Phi(S))\ge k+1$ implies the same for $B$.
At least one inequality holds by \eqref{eqn:loss-dichotomy}.
The equality cases show that the improvement is strict
unless the allocation is the corresponding one of $A,B$.

Finally, an allocation weakly improving both utilities of $A$
must satisfy $\ell(\Phi(S))=0$ and $\ell(S)\le k+1$.
Equations \eqref{eqn:loss-dichotomy} and
\eqref{eqn:first-PO-unique} force it to be $A$ itself.
Thus $A$ is PO, and symmetrically so is $B$.
\end{proof}

In $A$, agent $1$ has value $|U|-(k+1)$.
Let $D\subseteq P\cup R'$ with $|D|\le k$, and write
$a=|D\cap P|$ and $d=|D\cap R'|$.
The complement of $(P\cup R')\setminus D$ is
$R\cup\{q\}\cup D$, with counts $(a,1,k,d)$.
Since $a+d\le k$, \eqref{eqn:loss-function} gives
\[
 v_1((P\cup R')\setminus D)
       =|U|-k>|U|-(k+1)=v_1(R\cup\{q\}).
\]
Hence $A$ is not $\mathrm{EF}k$.
Interchanging the agents and applying $\sigma$ shows that
$B$ is not $\mathrm{EF}k$ either.
By \fullref{clm:po-structure}, no \efkpo\ allocation exists.
\end{proof}

\subsection{\texorpdfstring{$\mathrm{EF}k+\mathrm{PO}$}{EFk+PO} Allocations for Monotone Subadditive Valuations}

We prove part~\ref{itm:subadditive} of \fullref{thm:efkpo}
by giving matching existence and nonexistence bounds.

\begin{lemma}\label{lem:subadditive-existence}
Every instance with two agents, at most $2k+1$ goods, and
monotone subadditive valuations admits an \efkpo\ allocation.
\end{lemma}

\begin{proof}
If $|M|\le k$, every allocation is $\mathrm{EF}k$,
so any $\mathrm{PO}$ allocation suffices.
Assume henceforth that $|M|\ge k+1$.

Let $S$ be a maximum-value $k$-element subset for agent $1$.
The allocation $X=(S,M\setminus S)$ is $\mathrm{EF}k$:
agent $2$ can delete all of $S$, while agent $1$ can delete
at most $k$ goods from the other bundle to leave at most
one good, whose value is at most $v_1(S)$.

Choose a $\mathrm{PO}$ allocation $Y=(Y_1,Y_2)$ satisfying
$v_1(Y_1)\ge v_1(S)$ and
$v_2(Y_2)\ge v_2(M\setminus S)$.
Such an allocation exists since there are finitely many
allocations. We show that $Y$ is $\mathrm{EF}k$.

Let $T=Y_1\cap S$. Then $|T|\le k$ and
$Y_1\setminus T\subseteq M\setminus S$, so
$v_2(Y_1\setminus T)\le v_2(Y_2)$.
Thus agent $2$ is $\mathrm{EF}k$ towards agent $1$.

If $Y_1\ne\varnothing$, then $|Y_2|\le2k$.
Choose $D\subseteq Y_2$ with $|D|\le k$ and
$|Y_2\setminus D|\le k$.
By the choice of $S$ and monotonicity,
$v_1(Y_1)\ge v_1(S)\ge v_1(Y_2\setminus D)$.

If $Y_1=\varnothing$, then $v_1(S)=0$.
The choice of $S$ and monotonicity imply that every subset
of size at most $k$ has value zero for agent $1$.
In particular, every singleton has value zero, so
subadditivity gives $v_1(M)=0$.
Agent $1$ is therefore trivially $\mathrm{EF}k$ towards
agent $2$.
\end{proof}

The bound is tight even for identical valuations.

\begin{lemma}\label{lem:subadditive-example}
There exists an instance with two agents, $2k+2$ goods,
and identical monotone subadditive valuations that admits
no \efkpo\ allocation.
\end{lemma}

\begin{proof}
Let $M=\{a\}\cup B$, where
$B=\{b_1,\ldots,b_{2k+1}\}$.
Both agents have valuation
$v(S)=\max\{k\mathbf{1}[a\in S],|S\cap B|\}$,
which is monotone and subadditive.

The only $\mathrm{PO}$ allocations are
$(\{a\},B)$ and $(B,\{a\})$.
Indeed, consider $X=(X_1,X_2)$ and assume without loss of
generality that $a\in X_1$.
If $|X_1\cap B|\le k$, then $(\{a\},B)$ weakly improves
both utilities, strictly unless $X=(\{a\},B)$.
If $|X_1\cap B|>k$, then $|X_2\cap B|\le k$, and
$(B,\{a\})$ strictly Pareto dominates $X$.
The two displayed profiles, $(k,2k+1)$ and $(2k+1,k)$,
are incomparable, so these are exactly the PO allocations.

In either allocation, the agent receiving $\{a\}$ has value
$k$, but values the other bundle at least $k+1$ after
any $k$ deletions. Thus neither allocation is $\mathrm{EF}k$.
\end{proof}

\begin{remark}
The valuation $v$ is in fact monotone $\mathrm{XOS}$.
Let $v_1,v_2$ be additive valuations with
$v_1(a)=k$, $v_1(b_i)=0$, $v_2(a)=0$, and $v_2(b_i)=1$
for every $b_i\in B$.
Then $v(S)=\max\{v_1(S),v_2(S)\}$.
\end{remark}


\section{\texorpdfstring{$\mathrm{NP}$}{NP}-hardness of \texorpdfstring{$\mathrm{EF}k+\mathrm{PO}$}{EFk+PO} for Submodular Valuations}
\label{sec:hardness-goods}

We prove \fullref{thm:efonepo-NPhard} for every fixed integer
$k\ge1$. This also implies hardness when $k$ is part of the input.

\EFkPOHardness*

\begin{proof}[{Proof of \fullref{thm:efonepo-NPhard}}]
Fix an integer $k\ge1$. We reduce from $\bvc$, which is
$\mathrm{NP}$-complete by \cite[Lemma~1]{ConitzerS06}.

\begin{definition}[$\bvc$]\label{def:halfvc}
Given a graph on an even number $n$ of vertices, determine
whether it admits a vertex cover of size at most $n/2$.
\end{definition}

Let $G_0$ be an instance on $n$ vertices.
Form $G=(V,E)$ by adding a disjoint clique on $k+2$ vertices.
Then $|V|=n+k+2$, and
\[
 G_0\text{ has a vertex cover of size at most }n/2
 \quad\Longleftrightarrow\quad
 G\text{ has one of size at most }n/2+k+1.
\]
Indeed, the added clique requires exactly $k+1$ vertices
in a minimum cover. It also guarantees that $|E|>0$.

Let $v_1,v_2$ be the unweighted coverage valuations from
\fullref{lem:efk-po-example}, on the goods
$M=\{p_1,\ldots,p_{2k+1}\}\cup\{q\}\cup R\cup R'$,
where $|R|=|R'|=k$. Put $U=v_1(M)=v_2(M)$.
By \fullref{clm:po-structure}, their only $\mathrm{PO}$
allocations are
\[
 \mathcal A=
 \bigl(\{q\}\cup R,\{p_1,\ldots,p_{2k+1}\}\cup R'\bigr),
 \qquad
 \mathcal B=
 \bigl(\{p_1,\ldots,p_{2k+1}\}\cup R,\{q\}\cup R'\bigr),
\]
with utility profiles $(U-k-1,U)$ and $(U,U-k-1)$,
respectively. At either allocation, deleting at most $k$
goods from the bundle of the agent with value $U$ leaves
value exactly $U-k$ for the other agent.

Every good has positive marginal value within its assigned
bundle at either displayed allocation, as follows directly
from the coverage construction. Consequently, these remain
the only $\mathrm{PO}$ assignments of $M$ when a third agent
values $M$ at zero: complete any partial two-agent split and
then weakly improve it to $\mathcal A$ or $\mathcal B$.
Unless the original assignment was already that exact split,
some improvement is strict.

Introduce a vertex good $g_v$ for every $v\in V$, and let
$P=\{g_v:v\in V\}$. For $Q\subseteq P$, let $\vcov(Q)$
count the edges having an endpoint whose good belongs to $Q$.
The constructed instance $\mathcal I$ has three agents
$a_1,a_2,a_3$, with valuations
\[
\begin{aligned}
 w_1(X)&=|E|\,v_1(X\cap M)+\vcov(X\cap P),\\
 w_2(X)&=|E|\,v_2(X\cap M),\\
 w_3(X)&=|X\cap P|.
\end{aligned}
\]
These are unweighted coverage valuations: make $|E|$ copies
of each element in the descriptions of $v_1,v_2$; for agent 1,
additionally let each vertex good cover its incident edges
in a disjoint universe. Agent 3 simply counts vertex goods,
so $w_3$ is binary additive. Since $k$ is fixed, the construction
has polynomial size.

In every $\mathrm{PO}$ allocation, the goods in $M$ are assigned
as in $\mathcal A$ or $\mathcal B$, since rearranging them
affects no other contribution. Moreover, agent 2 receives no
vertex good: moving one to agent 3 would be a Pareto improvement.
The assignment $\mathcal B$ cannot be $\mathrm{EF}k$, since
agent 2 has value $|E|(U-k-1)$ but values agent 1's bundle
at least $|E|(U-k)$ after any $k$ deletions.
Thus every \efkpo\ allocation uses $\mathcal A$.

\begin{claim}\label{clm:forward}
If $G_0$ has a vertex cover of size at most $n/2$, then
$\mathcal I$ admits an \efkpo\ allocation.
\end{claim}

\begin{proof}
Then $G$ has a vertex cover of size at most $n/2+k+1$.
Choose a minimum vertex cover $S$, and let
$S'=\{g_v:v\in S\}$. Consider
\[
 X_1=\{q\}\cup R\cup S',\qquad
 X_2=\{p_1,\ldots,p_{2k+1}\}\cup R',\qquad
 X_3=P\setminus S'.
\]
Agent 1 has value $|E|(U-k)$, exactly its value for $X_2$
after deleting at most $k$ goods. It does not envy agent 3,
whose bundle has value at most $|E|$ for it.
Agent 2 is at its maximum value and envies nobody.
Agent 3 values $X_2$ at zero, and its $\mathrm{EF}k$
comparison with agent 1 holds because
\[
 |V|-|S|=n+k+2-|S|\ge\max\{0,|S|-k\}.
\]
Hence $X$ is $\mathrm{EF}k$.

To prove PO, consider an allocation weakly improving every
agent. Agent 2 must retain value $|E|U$.
Since agent 1's edge contribution is at most $|E|$, its
value from $M$ must remain at least $|E|(U-k-1)$.
The properties above therefore force exactly the
$\mathcal A$ assignment of $M$.
Agent 1 must consequently still receive a vertex cover,
requiring at least $|S|$ vertex goods by the minimality of $S$.
Preserving agent 3's utility permits at most $|S|$ vertex goods
outside its bundle. Thus all three utilities remain unchanged,
proving PO.
\end{proof}

\begin{claim}\label{clm:reverse}
If $\mathcal I$ admits an \efkpo\ allocation, then $G_0$
has a vertex cover of size at most $n/2$.
\end{claim}

\begin{proof}
Let $X$ be an \efkpo\ allocation, and let $S$ be the vertices
whose goods agent 1 receives. Its assignment of $M$ is
$\mathcal A$, so agent 1's $\mathrm{EF}k$ comparison with
agent 2 gives
\[
 |E|(U-k-1)+\vcov(\{g_v:v\in S\})\ge |E|(U-k).
\]
Thus $S$ covers every edge of $G$.
All remaining vertex goods belong to agent 3, whose
$\mathrm{EF}k$ comparison with agent 1 gives
\[
 n+k+2-|S|\ge\max\{0,|S|-k\}.
\]
Consequently $|S|\le n/2+k+1$.
Since the added clique requires at least $k+1$ vertices,
$G_0$ has a vertex cover of size at most $n/2$.
\end{proof}

Claims~\ref{clm:forward} and~\ref{clm:reverse} establish the
reduction and complete the proof.
\end{proof}


\bibliographystyle{alpha}
\bibliography{ref_arxiv_v2}

\newpage 

\begin{appendices}

\section{Existence of \texorpdfstring{$\mathrm{EF1}+\mathrm{PO}$}{EF1+PO} for Identical Monotone Submodular Valuations}
\label{sec:identical}

We record a proof that $\mathrm{EF1}+\mathrm{PO}$ allocations always
exist when all agents have identical monotone submodular valuations.
The proof uses the transferability property of monotone submodular
valuations from \cite[Lemma~3.5]{BenabbouCIZ21}.

\begin{theorem}[Identical monotone submodular valuations]
\label{thm:identical-submodular-ef1-po}
Suppose all agents have the same monotone submodular valuation $v$.
Then, for any number of agents, an $\mathrm{EF1}+\mathrm{PO}$ allocation
exists. Indeed, every leximin allocation has both properties.
\end{theorem}

\begin{proof}
Let $X$ be a leximin allocation \cite{PlautR20}, namely, one whose
nondecreasing utility vector
$\operatorname{sort}(v(X_1),\ldots,v(X_n))$ is lexicographically
maximal. It is Pareto optimal, since a Pareto improvement would
strictly improve this sorted vector.

Suppose $X$ is not $\mathrm{EF1}$. Then there are agents $i,j$ such
that, writing $r=v(X_i)$, we have $v(X_j\setminus\{g\})>r$ for every
$g\in X_j$. In particular, $v(X_j)>v(X_i)$. By the transferability
property \cite[Lemma~3.5]{BenabbouCIZ21}, some $g\in X_j$ satisfies
$v(X_i\cup\{g\})>r$. Indeed, otherwise diminishing marginal returns
would imply $v(X_i\cup X_j)=v(X_i)$, contradicting monotonicity.

Transfer $g$ from $j$ to $i$, obtaining an allocation $Y$. Both
$v(Y_i)$ and $v(Y_j)$ are strictly greater than $r$, and all other
utilities are unchanged. Thus no utility below $r$ changes, one
occurrence of $r$ disappears, and no new utility at most $r$ is
created. The sorted utility vector strictly improves, contradicting
the choice of $X$.
\end{proof}

This contrasts with monotone \emph{subadditive} valuations:
even two agents with identical valuations over four goods need not
admit an $\mathrm{EF1}+\mathrm{PO}$ allocation
\cite[Theorem~3.3]{CaragiannisKMPSW19}.


\section{Nonexistence of \texorpdfstring{$\mathrm{EF1}+\mathrm{PO}$}{EF1+PO} for Chores under Coverage Costs}
\label{sec:chores}

We show that the instance in \fullref{tab:valuations} also admits no
$\mathrm{EF1}+\mathrm{PO}$ allocation when the items are chores.
We use the same functions $v_i$ as costs, so agents prefer bundles
of lower value. Recall that $\mathrm{EF1}$ now allows the removal
of one chore from the envious agent's own bundle.

\begin{theorem}
\label{thm:chores}
The instance in \fullref{tab:valuations} admits no
$\mathrm{EF1}+\mathrm{PO}$ allocation when the values are interpreted
as chore costs.
\end{theorem}

\begin{proof}
If one bundle is empty, the allocation is Pareto optimal but not
$\mathrm{EF1}$: the agent receiving all chores retains positive
cost after any one deletion.

Henceforth suppose both bundles are nonempty. The two cost
functions differ only by interchanging $r$ and $r'$. For agent $i$,
write $r_{\min}^i$ for the cheaper of these two chores and
$r_{\max}^i$ for the other. Their costs are $2$ and $4$,
respectively, and
$r_{\min}^i=r_{\max}^{3-i}$.
We call $p_1,p_2,p_3$ the $p$-chores.

\begin{claim}
\label{clm:valuebelow5}
No allocation gives both agents cost strictly less than $5$.
\end{claim}

\begin{proof}
Some agent receives at least two $p$-chores. These have cost $4$,
and adding any other chore raises the cost to at least $5$.
Thus, if both costs were below $5$, this agent would receive
exactly two $p$-chores. The other agent would then receive
$q,r,r'$, whose cost is $5$.
\end{proof}

\begin{claim}
\label{clm:value6}
If one agent has cost at most $2$, the other has cost at least $6$.
Moreover, a cost profile $(2,6)$ or $(6,2)$ is possible only when
the agent with cost $2$ receives $\{p_1\}$ or $\{p_2\}$.
\end{claim}

\begin{proof}
An empty bundle leaves all chores, of cost $7$, to the other agent.
Every nonempty bundle of cost at most $2$ is a singleton
$\{p_1\}$, $\{p_2\}$, $\{p_3\}$, or $\{r_{\min}^i\}$.
The corresponding complementary bundles have cost $6$, $6$,
$7$, and $7$, respectively, for the other agent.
\end{proof}

Thus $(0,7)$, $(7,0)$, $(2,6)$, and $(6,2)$ are attainable
Pareto-optimal cost profiles. The first two fail $\mathrm{EF1}$,
as already noted. For either of the latter two, the bundle of
cost $2$ is $\{p_1\}$ or $\{p_2\}$. The other bundle contains two
$p$-chores and $q$. After any one deletion, it still contains
either two $p$-chores or $q$ and a $p$-chore, of cost at least $4$.
Its owner therefore still envies the singleton bundle, which
it values at $2$.

Any other Pareto-optimal profile must have a coordinate equal
to $5$. Indeed, a profile with a coordinate $7$ and the other
positive is dominated by $(7,0)$ or $(0,7)$. A profile with a
coordinate $6$ and the other greater than $2$ is dominated by
$(6,2)$ or $(2,6)$. The remaining possibility follows from
\fullref{clm:valuebelow5}.

\begin{claim}
\label{clm:value5}
If one agent has cost $5$, the other has cost at least $4$.
The profiles $(5,4)$ and $(4,5)$ occur precisely when the bundles
are $\{p_1,p_2\}$ and $\{p_3,q,r,r'\}$.
\end{claim}

\begin{proof}
Suppose agent $i$ has cost $5$ and agent $j=3-i$ has cost below
$4$. By \fullref{clm:value6}, agent $j$ must have cost exactly $3$.
Its bundle is therefore $\{q\}$ or
$\{p_3,r_{\min}^j\}$. In the first case, agent $i$ receives all
three $p$-chores, of cost $6$. In the second, it receives
$p_1,p_2,q,r_{\min}^i$, again of cost at least $6$.
Both cases are impossible.

It remains to characterize profiles $(5,4)$ and $(4,5)$.
Neither bundle can contain all three $p$-chores, so one bundle
contains exactly two.

If the bundle of cost $4$ contains two $p$-chores, it contains
nothing else. If one of them is $p_3$, the complementary bundle
contains $p_1$ or $p_2$ together with $q,r,r'$, of cost $6$.
Hence this bundle must be $\{p_1,p_2\}$, leaving
$\{p_3,q,r,r'\}$, of cost $5$, to the other agent.

Now suppose the bundle of cost $5$ contains two $p$-chores.
It must contain at least one of $r,r'$, since otherwise both
would lie in the bundle of cost $4$, although together they
cost $5$. Its two $p$-chores cannot be $p_1,p_2$, since adding
either $r$ or $r'$ gives cost at least $6$. They must therefore
be $p_\ell,p_3$ for some $\ell\in\{1,2\}$.

Let $i$ own this bundle. Adding $r_{\max}^i$ would give cost
$6$, so it must contain $r_{\min}^i$. Adding $q$ would also give
cost $6$. Thus the bundle is exactly
$\{p_\ell,p_3,r_{\min}^i\}$. The other bundle contains
$q,r_{\min}^j$, whose cost for agent $j$ is $5$, contradicting
its cost $4$.
\end{proof}

Every remaining profile is dominated by $(5,4)$ or $(4,5)$.
Finally, neither allocation realizing these profiles is
$\mathrm{EF1}$. For either agent, the bundle
$\{p_3,q,r,r'\}$ covers five universe elements, each at least
twice. Removing any one chore therefore leaves its cost $5$,
whereas $\{p_1,p_2\}$ has cost $4$.
\end{proof}


\section{\texorpdfstring{$\mathrm{NP}$}{NP}-hardness of \texorpdfstring{$\mathrm{EF1}+\mathrm{PO}$}{EF1+PO} for Chores}
\label{sec:NP-hardness-chores}

Hosseini, Narang, and W\k{a}s
\cite[Proposition~4.2]{HosseiniNW25} proved that deciding whether
an $\mathrm{EF1}+\mathrm{PO}$ allocation exists is $\mathrm{NP}$-hard
for identical costs arising from an unweighted rooted tree.
These are unweighted coverage costs: each chore covers the edges
on its path to the root. Their reduction from
$3$-\textsc{Partition} uses a variable number of agents.
We give a reduction from $\bvc$ for three agents, two of whom
have identical costs.

\begin{theorem}
\label{thm:NP-hardness-chores}
Deciding whether an $\mathrm{EF1}+\mathrm{PO}$ allocation exists
is $\mathrm{NP}$-hard for three agents with unweighted coverage
chore costs.
\end{theorem}

\begin{proof}
Let $G=(V,E)$ be a $\bvc$ instance with an even number $n$ of
vertices and $m=|E|\ge1$ edges. The question is whether $G$ has
a vertex cover of size at most $n/2$.

For each vertex $v$, introduce two chores
$P_v=\{p_v^1,p_v^2\}$, and introduce three further chores
$c,c_1,c_2$.

For a bundle $X$, let $S(X)=\{v:X\cap P_v\neq\varnothing\}$.
For $T\subseteq V$, let $v_{\operatorname{cov}}(T)$ count the
edges having an endpoint in $T$. Set $L=(n+1)m$ and define
\begin{align*}
v_1(X)
&=L\cdot\mathbf{1}[X\cap\{c_1,c_2\}\neq\varnothing]
 +(n+1)v_{\operatorname{cov}}(S(X))
 +|S(X)|+L\cdot\mathbf{1}[c\in X],\\
v_2(X)=v_3(X)
&=(n+1)\cdot\mathbf{1}[X\cap\{c_1,c_2\}\neq\varnothing]
 +|S(X)|.
\end{align*}
These are unweighted coverage costs. The term $|S(X)|$ uses
one universe element per vertex, covered by both chores in
$P_v$. For agent $1$, each edge contributes $n+1$ additional
elements, covered by the chores of its endpoints. Each remaining
term uses a separate set of elements, covered by the specified
special chores. There are $2n+3$ chores and $O(mn+n)$ universe
elements, so the construction has polynomial size.

We first describe every $\mathrm{EF1}+\mathrm{PO}$ allocation
$X$. Pareto optimality assigns $c$ to agent $2$ or $3$: moving
it from agent $1$ to either agent strictly lowers agent $1$'s
cost and changes no other cost. Since agents $2$ and $3$ are
identical, suppose $c\in X_3$.

Pareto optimality also keeps each pair $P_v$, and the pair
$\{c_1,c_2\}$, together. If a pair were split, moving one chore
to the holder of the other would add no cost to the recipient
and strictly lower the donor's cost. Moreover, $\{c_1,c_2\}$
must go to agent $1$. Otherwise its holder retains cost at
least $n+1$ after any one deletion, while assigning each other
bundle cost at most $n$, contradicting $\mathrm{EF1}$.

Hence deleting any chore from its owner's bundle does not
change that owner's cost: every pair stays together, and $c$
costs agent $3$ zero. Write $S_i=S(X_i)$. For agent $1$ to
satisfy $\mathrm{EF1}$ towards agent $2$, we must have
\[
L\le v_1(X_1)\le v_1(X_2)
  =(n+1)v_{\operatorname{cov}}(S_2)+|S_2|.
\]
If $S_2$ missed an edge, the right-hand side would be at most
$(n+1)(m-1)+n=L-1$. Thus $S_2$ is a vertex cover.

For agents $2$ and $3$ to satisfy $\mathrm{EF1}$ towards each
other, we must have $|S_2|=|S_3|$. These sets are disjoint,
so $|S_2|\le n/2$. Thus an $\mathrm{EF1}+\mathrm{PO}$ allocation
gives a vertex cover of size at most $n/2$.

Conversely, suppose $G$ has a vertex cover of size at most
$n/2$. Enlarge it to a cover $S$ of size exactly $n/2$, and set
\[
X_1=\{c_1,c_2\},\qquad
X_2=\bigcup_{v\in S}P_v,\qquad
X_3=\{c\}\cup\bigcup_{v\notin S}P_v.
\]
The agents' costs are $(L,n/2,n/2)$. Agent $1$ assigns each
other bundle cost at least $L$: $S$ covers every edge, and
$X_3$ contains $c$. Agents $2$ and $3$ assign each other's
bundle cost $n/2$, and assign $X_1$ cost $n+1$.
Thus the allocation is envy-free, and hence $\mathrm{EF1}$.

Finally, consider an allocation weakly lowering all three
costs. The cost bounds $n/2$ for agents $2$ and $3$ prohibit
either from receiving $c_1$ or $c_2$. Both chores therefore
remain with agent $1$. Its cost bound $L$ then prohibits it
from receiving any other chore.

All vertex chores must therefore go to agents $2$ and $3$.
For each vertex, at least one of these agents incurs its
unit cost, so their total cost is at least $n$. Their
individual bounds force both costs to remain $n/2$, while
agent $1$'s cost remains $L$. No cost can strictly fall,
proving Pareto optimality.
\end{proof}


\end{appendices}

\end{document}